\documentclass[aps,pra,twocolumn,amsmath,groupedaddress,nofootinbib]{revtex4-1}

\usepackage{amsmath,bbm}
\usepackage{graphicx}
\usepackage{amsfonts,mathpazo,palatino}
\usepackage{amssymb}
\usepackage{amsmath, amssymb, amsthm, verbatim, changes}
\usepackage{color,xcolor,longtable}

\definecolor{myurlcolor}{rgb}{0,0,0.7}
\definecolor{myrefcolor}{rgb}{0,0.7,0}
\usepackage[unicode=true,pdfusetitle, bookmarks=true,bookmarksnumbered=false,bookmarksopen=false, breaklinks=false,pdfborder={0 0 0},backref=false,colorlinks=true, linkcolor=myrefcolor,citecolor=myurlcolor,urlcolor=myurlcolor]{hyperref}

\theoremstyle{plain}
\newtheorem{thm}{Theorem}

\newtheorem{obse}[thm]{Observation}

\theoremstyle{definition}

\theoremstyle{remark}

\newcommand{\sr}[1]{\langle[#1]\rangle}

\begin{document}


\title{Simple and tight monogamy relations for a class of Bell inequalities}



\author{Remigiusz Augusiak} \email{augusiak@cft.edu.pl}
\affiliation{Center for Theoretical Physics, Polish Academy of Sciences, Aleja Lotnik\'ow 32/46, 02-668 Warsaw, Poland}




\begin{abstract}
Physical principles constraints the way nonlocal correlations can be 
distributed among distant parties in a Bell-type experiment. These
constraints are usually expressed by monogamy relations that bound 
the amount of Bell inequality violation observed by a set of parties 
by the violation observed by a different set of parties.
Here we show that the no-signaling principle yields
a simple and tight monogamy relations for an important class of bipartite 
and multipartite Bell inequalities. We also link these trade-offs to the 
guessing probability--a key quantity in the device-independent information processing.
%
\end{abstract}


\maketitle

\section{Introduction}

Consider spatially separated parties sharing some physical 
system and assume that they perform measurements 
on their shares of the system. It was Bell who proved that 
in some situations correlations between outcomes of these 
measurements cannot be explained by means of local hidden
variable models \cite{Bell} (see also Ref. \cite{reviewus}). 
Such correlations are termed nonlocal 
and one usually detects them with the aid of 
Bell inequalities (see, e.g., Refs. \cite{review}). And, importantly, 
they have recently become 
a key resource for device-independent quantum information
tasks. In particular, they allow for security not achievable with classical resources \cite{key1,key2}, 
outperform classical correlations at communication complexity problems \cite{compl}, 
or are crucial for generation \cite{gen1,gen2} and amplification \cite{ampl1} of true randomness.

However, it turns out that such nonlocal correlations 
cannot be distributed between distant parties in an arbitrary way. 
In fact, physical principles impose certain constraints on 
the way these resources can be shared by distant parties.
These constraints are usually referred to as monogamy relations and 
are important from the point of view of cryptographic security (cf. Ref. \cite{Leverrier}). 
In particular, in Refs. \cite{Barrett,Masanes} it was shown 
that in any theory respecting the no-signaling principle, 
which prevents any faster-than-light communication among the parties,
two observers sharing extremal nonlocal correlations 
in the polytope of all nonsignaling correlations must remain 
uncorrelated to any other party. Using the concept of linear programming, 
this statement was later put on a quantitative ground by Toner \cite{Toner}. He 
showed that if three parties $A$, $B$ and $C$ share nonsignaling correlations, then 
the value of the Clauser-Horne-Shimony-Holt (CHSH) Bell 
expression \cite{CHSH} observed by two different pairs of parties, say $AB$ and $AC$, 
satisfy
%
$I_{AB}+I_{AC}\leq 4,$
%
where the CHSH Bell inequality is defined as
\begin{equation}\label{CHSH}
I_{AB}:=\langle A_1B_1\rangle+\langle A_1B_2\rangle+
\langle A_2B_1\rangle-\langle A_2B_2\rangle\leq 2.
\end{equation}
with $A_i$ and $B_i$ denoting dichotomic observables measured
by the parties $A$ and $B$. This, in particular, 
implies that only a single pair, $AB$ or $AC$
can violate (\ref{CHSH}). This
monogamy relation is tight in the sense that 
for any pair $I_{AB}$ and $I_{AC}$ saturating it, one 
is able to construct nonsignaling correlations realizing these values.
If then correlations are to obey the laws of quantum theory 
(which is a particular example of a no-signaling theory)
a stronger monogamy relation in terms of
the CHSH Bell inequality can be derived \cite{TonerVerstraete}.


Following the idea of symmetric extensions and shareability of correlations, 
Toner's monogamy relation was further generalized to any bipartite 
Bell inequality (involving any number of measurements and outcomes) 
\cite{Marcin}. 
These general monogamy relations are however quite complicated and
alternative constructions of simpler 
relations were then put forward in Refs. \cite{RaviMon1,Ravi2}, 
which exploit the concept of the contradiction number and the graph theory, respectively.
However, both methods yield monogamy relations that are in general not tight.


The aim of this note is to introduce simple
monogamy relations for an important class of Bell inequalities.
Importantly, our inequalities are tight irrespectively of the scenario considered, and thus improve
on the results of Ref. \cite{RaviMon1}.
We will also link the new trade-offs to the guessing probability---a central quantity to the device-independent 
information processing---showing
that in certain cases our monogamy relations impose tight bounds on the latter.




\section{Preliminaries}

Before presenting our main results we need to introduce some 
basic notions and terminology that is necessary for further considerations, 
and also recall some of the previous constructions of monogamy 
relations.


\subsection{Bell scenario}
\label{Sec:Preliminaries2}
Let us begin by stating the scenario.
We consider the usual multipartite Bell-type experiment in which 
$N$ spatially separated parties $A_i$ $(i=1,\ldots,N)$ 
share some physical system. Each party can perform one of $M$ 
measurements on their share of the system, and each measurement is assumed to have 
$d$ outcomes. For the party $A_i$ we denote the measurement choices and outcomes 
by $x_i=1,\ldots,M$ and $a_i=0,\ldots,d-1$, respectively. 
One usually refers to the this scenario as to $(N,M,d)$ scenario.


The above experiment, when repeated many times, creates correlations
that are described by a collection of conditional probabilities 
\begin{equation}
\{p(a_1,\ldots,a_N|x_1,\ldots,x_N)\}_{a_1,\ldots,a_N;x_1,\ldots,x_N},
\end{equation}
where $p(a_1,\ldots,a_N|x_1,\ldots,x_N)$ is
the probability that the party $A_i$ obtains $a_i$ upon performing 
the $x_i$th measurements. In what follows we will also use $p(\mathbf{a}|\mathbf{x})$ with $\mathbf{a}=a_1,\ldots, a_N$ and 
$\mathbf{x}=x_1,\ldots, x_N$ to denote these probabilities, and 
$\mathbf{p}$ to denote the set $\{p(\mathbf{a}|\mathbf{x})\}$.
%
%
%
%
In the case of three parties $(N=3)$ we utilize a slightly different notation in which 
the parties are denoted by $A$, $B$, and $C$, while their measurement choices and results by $x,y,z$ and $a,b,c$, 
respectively. In this notation, $p(abc|xyz)$
is the probability of obtaining $a,b,c$ upon measuring $x,y,z$.

We consider correlations $\{p(\mathbf{a}|\mathbf{x})\}$
that obey the no-signaling principle, which prevents any faster-than-light communication among the parties. It states that the outcomes obtained by a group of parties cannot depend on the measurement choices made by the remaining parties. Mathematically, this is formulated as the following set of
linear conditions for the probabilities $p(\mathbf{a}|\mathbf{x})$:
\begin{eqnarray}\label{nosignaling}
&&\sum_{a_{i}}p(a_1,\ldots,a_{i},\ldots,a_{i_N}|x_1,
\ldots,x_{i},\ldots, x_N)\nonumber\\
&&=\sum_{a_{i}}p(a_1,\ldots,a_{i},\ldots,a_N|x_1,
\ldots,x_{i}',\ldots, x_N)
\end{eqnarray}
%
%
which hold for all $x_i,x_i'$ and $a_1,\ldots,a_{i-1},a_{i+1},\ldots,a_N$ 
and $x_1,\ldots,x_{i-1},x_{i+1},\ldots,x_N$ and all $i=1,\ldots,N$. Notice that 
in a given scenario correlations obeying the no-signaling principle form a polytope---a 
bounded convex set with finite number of extremal elements.

\subsection{Bell inequality and previous monogamy relations}
\label{Sec:Preliminaries3}
Let us now state the Bell inequality which we will use to formulate our 
monogamy relations. For clarity we begin from the bipartite case. 

The Bell inequality we consider was introduced by Barrett, Kent and Pironio (BKP) in Ref. \cite{BKP} and is given by
\begin{equation}\label{BKPNd}
I^{M,d}_{AB}:=\sum_{\alpha=1}^{M}\left(\sr{A_{\alpha}-B_{\alpha}}
+\sr{B_{\alpha}-A_{\alpha+1}}\right)\geq
d-1,
\end{equation}
where $\langle \Omega\rangle=\sum_{i=1}^{d-1}P(\Omega=i)$ is the standard expectation value of the variable $\Omega$, $[X]$ denotes $X$ modulo $d$, and we use the convention that $X_{M+1}=X_1+1$. For $M=2$
this class of Bell inequalities reproduces the 
Collins-Gisin-Linden-Massar-Popescu (CGLMP) 
Bell inequalities introduced in Ref. \cite{CGLMP}, while for $d=2$
the so-called chained Bell inequalities \cite{chainedBI}, and in the particular case of
$M=d=2$ the CHSH Bell inequality (\ref{CHSH}).

For further purposes 
recall that the maximal violation of the BKP Bell inequality by no-signaling correlations 
amounts to $I_{AB}^{M,d}=0$ and correlations
realizing this value, denoted $\mathbf{p}_{2}^{\mathrm{NL}}=\{p^{\mathrm{NL}}(a,b|x,y)\}$, are given by
%
%
\begin{eqnarray}
p^{\mathrm{NL}}(a,a|x,x)&\!\!=\!\!&\frac{1}{d},\quad x=1,\ldots,M\\
p^{\mathrm{NL}}(a,a|x+1,x)&\!\!=\!\!&\frac{1}{d},\quad x=1,\ldots,M-1\\
p^{\mathrm{NL}}(a,a+1|1,M)&\!\!=\!\!&\frac{1}{d},
\end{eqnarray}
which implies that $p^{\mathrm{NL}}(a,b|x,x)=0$ for $a\neq b$ and $x=1,\ldots,M$,
$p^{\mathrm{NL}}(a,b|x+1,x)=0$ for $a\neq b$ and $x=1,\ldots,M-1$, and $p^{\mathrm{NL}}(a,b|1,M)=0$
for $b\neq a+1$. For the remaining choices of measurements $x$
and $y$ one simply takes $p^{\mathrm{NL}}(a,b|x,y)=1/d^2$. 
It is not difficult to see that 
this probability distribution 
satisfies all the constraints (\ref{nosignaling}) as both its reductions 
are uniform, i.e.,
$p_A^{\mathrm{NL}}(a|x)=p_B^{\mathrm{NL}}(b|y)=1/d$ for any $a,b,x,y$. 
Let us also mention that  
the local correlations saturating the inequality (\ref{BKPNd}), in what follows
denoted $\mathbf{p}_2^{\mathrm{loc}}=\{p^{\mathrm{loc}}(a,b|x,y)\}$, read
\begin{equation}\label{lok1}
p^{\mathrm{loc}}(a,b|x,y)=p_A^{\mathrm{loc}}(a|x)p_B^{\mathrm{loc}}(b|y)
\end{equation}
with the local probabilities defined as
%
$p_A(0|x)=p_B(0|y)=1$ for $x,y=1,\ldots, M.$
%

The BKP Bell inequality was generalized in Ref. \cite{rodrigo} to any number of 
observers $N$, and this generalization takes the following form
\begin{eqnarray}\label{BKPNMd}
I^{N,M,d}_{A_1\ldots
A_N}&\!\!\!:=\!\!\!&
\sum_{\boldsymbol{\alpha}}
(\sr{A_{\alpha_1}^{(1)}-A_{\alpha_1+\alpha_2-1}^{(2)}+\ldots
\nonumber\\
&&\hspace{0.75cm}+(-1)^{N}A_{\alpha_{N-2}+\alpha_{N-1}-1}^{(N-1)}+(-1)^{N-1}A_{\alpha_{N-1}}^{(N)}}\nonumber\\
&&\hspace{0.5cm}+\sr{A_{\alpha_1+\alpha_2-1}^{(2)}-A_{\alpha_1+1}^{(1)}+\ldots\nonumber\\
&&\hspace{1cm}+(-1)^{N-1}A_{\alpha_{N-2}+\alpha_{N-1}-1}^{(N-1)}+(-1)^{N}A^{(N)}_{\alpha_{N-1}}})\nonumber\\
&\geq&
M^{N-2}(d-1),
\end{eqnarray}
with $\boldsymbol{\alpha}=\alpha_1,\ldots,\alpha_{N-1}$, where each $\alpha_{i}=1,\ldots,M$. Clearly, for 
the particular case of two observers the above Bell inequality reproduces the one in (\ref{BKPNd}):
$I^{2,M,d}\equiv I^{M,d}$. It is also worth mentioning that it is of Svetlichny type \cite{Svetlichny}, 
that is, any correlations $\{p(\boldsymbol{a}|\boldsymbol{x})\}$ violating it are 
genuinely multipartite nonlocal (see Refs. \cite{GMN}).

The maximal violation of these Bell inequalities over all nonsignaling correlations
is again $I^{N,M,d}_{A_1\ldots A_N}=0$ and it is achieved by the probability distribution $\mathbf{p}^{\mathrm{NL}}_N=\{p^{\mathrm{NL}}(\mathbf{a}|\mathbf{x})\}$ with the probabilities $p^{\mathrm{NL}}(\mathbf{a}|\mathbf{x})$ defined in the following way
\begin{eqnarray}
&&p^{\mathrm{NL}}(\mathbf{a}|\alpha_1,\alpha_1+\alpha_2-1,\ldots,\alpha_{N-2}+\alpha_{N-1}-1,\alpha_{N-1})\nonumber\\
&&\hspace{1cm}=
\left\{
\begin{array}{cl}
\displaystyle\frac{1}{d^{N-1}}, & \displaystyle\sum_{i=0}^{N-1}(-1)^{i} a_{\alpha_i+\alpha_{i+1}-1}=f(\boldsymbol{\alpha})\\[2ex]
0, & \mathrm{otherwise}
\end{array}
\right.,
\end{eqnarray}
where it is assumed that $\alpha_0=\alpha_N=1$, $\alpha_i=1,\ldots,M$ for all $i=1,\ldots,N-1$
%
%
and the factor $f(\boldsymbol{\alpha})$ is defined as
\begin{equation}\label{f}
f(\boldsymbol{\alpha})=\sum_{i=0}^{N-1}(-1)^{i+1}H(\alpha_i+\alpha_{i+1}-1-M)
\end{equation}
with $H(x)$ being the Heaviside's function, i.e., $H(x)=1$ for $x>0$ and $x\leq 0$.
This factor is introduced
to take into account the fact that $A^{(j)}_{M+k}=A_{k}^{(j)}+1$ and the condition defining 
the probabilities in the Bell expression modifies. If for all $i=1,\ldots,N-2$, $\alpha_i+\alpha_{i+1}-1\leq M$, 
then $f=0$, but if for some $i$'s, $\alpha_i+\alpha_{i+1}-1>M$, then it might be that $f\neq 0$.

In the same way one has
\begin{eqnarray}
&&p^{\mathrm{NL}}(\textbf{a}|\alpha_1+1,\alpha_1+\alpha_2-1,\ldots,\alpha_{N-2}+\alpha_{N-1}-1,\alpha_{N-1})\nonumber\\
&&\hspace{1cm}=
\left\{
\begin{array}{ll}
\displaystyle\frac{1}{d^{N-1}}, \,\,& \displaystyle\sum_{i=0}^{N-1}(-1)^{i}a_{\alpha_i+\alpha_{i+1}-1}=f(\boldsymbol{\alpha})\\[2ex]
0, & \mathrm{otherwise}
\end{array}
\right.,
\end{eqnarray}
with $\alpha_0=2$, $\alpha_N=1$, $\alpha_i=1,\ldots,M$ for all $i=1,\ldots,N-1$, and $f(\boldsymbol{\alpha})$ 
is defined as in Eq. (\ref{f}), but with $\alpha_0=2$ (now the factor takes also into account that $\alpha_1+1$ exceeds $M$ for $\alpha_1=M$). 
For the remaining choices of the measurements we simply assume
%
$p^{\mathrm{NL}}(a_1,\ldots, a_N|\alpha_1,\ldots, \alpha_N)=1/d^N$.
%
%
On the other hand, the local correlations saturating the inequality (\ref{BKPNMd})
are straightforward generalization of those in Eqs. (\ref{lok1}), that is,
\begin{equation}
p_{A_1\ldots A_N}^{\mathrm{loc}}(\mathbf{a}|\mathbf{x})=\prod_{i=1}^Np_{A_i}^{\mathrm{loc}}(a_i|x_i)
\end{equation}
with the local probability distributions $\{p_{A_i}^{\mathrm{loc}}(a_i|x_i)\}$ defined as
%
$p_{A_i}^{\mathrm{loc}}(0|x_i)=1$
%
for any $x_i=1,\ldots,M$ and $i=1,\ldots,N$.

Now, having introduced the Bell inequality that we will use in further 
considerations, let us briefly recall the previous monogamy relations satisfied by it. 
The first monogamy relation is due to Paw\l{}owski and Brukner \cite{Marcin}
and reads
\begin{equation}\label{PB}
I_{AB_1}^{2,M,d}+I_{AB_2}^{2,M,d}+\ldots+I_{AB_M}^{2,M,d}\geq M(d-1).
\end{equation}
It implies that the sum of violations between the party $A$ and 
$M$ different Bobs $B_i$ cannod exceed $M$ times the classical bound
of the BKP Bell inequality. 
A much simpler monogamy relation for this Bell inequality was derived in Ref. \cite{RaviMon1} (see also Ref. \cite{Ravi2}) and it reads
\begin{equation}\label{monog15}
I_{AB}^{2,M,d}+I_{AC}^{2,M,d}\geq 2(d-1).
\end{equation}
Notice that for $M=2$ this monogamy is the same as 
the one in Eq. (\ref{PB}) and therefore it is tight. However, 
for any $M>2$ it is not tight except for the case 
$I_{AB}^{2,M,d}=I_{AC}^{2,M,d}=d-1$. Below we show that 
for other values of $I_{AB}^{2,M,d}$ and $I_{AC}^{2,M,d}$ that saturate 
the above inequality there does not exist nonsignaling correlations realizing these values. 

Our aim in what follows is to tighten the monogamy relation (\ref{monog15}) and then generalize it to any number of observers.

\section{Simple monogamy relations}

%
%
We are now ready to present our monogamy relations. For the sake of
clarity we will begin from the bipartite case, but, before that,  
let us show a simple proof of the monogamy relation 
(\ref{monog15}). 
\begin{obse}For any three-partite nonsignaling correlations
$\{p(a,b,c|x,y,z)\}$ with $M$ $d$-outcome measurements per party, 
the following inequality 
\begin{equation}\label{MonogNtight}
I_{AB}^{2,M,d}+I_{AC}^{2,M,d}\geq 2(d-1)
\end{equation}
is satisfied.
\end{obse}
\begin{proof}
Since for any random variable
$\Omega$ the identity $\sr{\Omega}+\sr{-\Omega-1}=d-1$ holds true,
one easily finds that
\begin{equation}
\sum_{j=2}^{M-1}(\sr{A_1-B_{j}}+\sr{B_j-A_1-1})-(M-2)(d-1)=0.
\end{equation}
By adding this expression to $I^{2,M,d}_{AB}$, the latter, after some simple movements,
can be rewritten in the following way
\begin{eqnarray}\label{BKPNd2}
I^{M,d}_{AB}&=&\sum_{i=1}^{M-1}(\sr{A_1-B_{i}}+\sr{B_i-A_{i+1}}\nonumber\\
&&\hspace{1cm}+\sr{A_{i+1}-B_{i+1}}+\sr{B_{i+1}-A_1-1})\nonumber\\
&&-(M-2)(d-1).
\end{eqnarray}
As a result the BKP Bell expression $I^{M,d}_{AB}$ is a combination of
$M-1$ CGLMP Bell expressions $I^{2,d}_{AB}$ involving different pairs of
$A$'s and $B$'s measurements. Clearly, the same holds for $I_{AC}^{M,d}$.
Plugging Eq. (\ref{BKPNd2}) into (\ref{MonogNtight}) one then has
\begin{eqnarray}\label{kadarka}
I_{AB}^{2,M,d}+I_{AC}^{2,M,d}&=&\sum_{i=1}^{M-1}\left(I_{AB}^{2,2,d}[A_1,A_{i+1};B_i,B_{i+1}]\right.\nonumber\\
&&\left.\hspace{1cm}+I_{AC}^{2,2,d}[A_1,A_{i+1};C_i,C_{i+1}]\right)\nonumber\\
&&-(M-2)(d-1),
\end{eqnarray}
where the square parentheses contain the measurement for which $I^{2,d}$
is defined. Now, it follows from Eq. (\ref{PB}) 
that $I_{AB}^{2,2,d}+I_{AC}^{2,2,d}\geq 2(d-1)$, which after being applied to 
(\ref{kadarka}) yields (\ref{MonogNtight}). This completes the proof.
\end{proof}

%
%
%
%
%
%

As already said, the trade-off (\ref{MonogNtight}) is in general 
not tight. The question we want to address now is whether 
it can be tightened by modifying it to the following form
\begin{equation}\label{General}
\alpha(M,d) I_{AB}^{M,d}+\beta(M,d)I_{AC}^{M,d}\geq \gamma(M,d)
\end{equation}
with $a(M,d), b(M,d)$ and $c(M,d)$ being some constants
that in general depend on $M$ and $d$. 
As proven in the following theorem this is the case for 
the BKP Bell inequality.


\begin{thm}\label{thm1}
All three-partite nonsignalling probability distributions $\{p(a,b,c|x,y,z)\}$
with $M$ inputs and $d$ outputs per site must satisfy the following pair of
inequalities:
%
\begin{eqnarray}\label{ThreeMonogamy}
\label{nier1} (M-1)I_{AB}^{M,d}+I^{M,d}_{AC}\geq M(d-1), \\
\label{nier2} I_{AB}^{M,d}+(M-1)I^{M,d}_{AC}\geq M(d-1).
\end{eqnarray}
%
%
%
%
%
\end{thm}
Before proving this theorem let us comment that the first inequality 
imposes a stronger constraint when $I^{M,d}_{AB}\leq I^{M,d}_{AC}$, while the second one 
when $I^{M,d}_{AB}\geq  I^{M,d}_{AC}$.
%
%
\begin{proof}In what follows we will prove the first inequality. The second one
will follow from the first one by exchanging $B\leftrightarrow C$.

Let us begin by introducing the following notation
\begin{equation}\label{def}
E_{XY}^{(j)}=\sr{X_j-Y_j}+\sr{Y_j-X_{j+1}},
\end{equation}
with $j=1,\ldots,M$, where we remember that 
$E_{XY}^{(M)}=\sr{X_M-Y_M}+\sr{X_M-Y_{1}-1}$.

Now, let us concentrate on the left-hand side of Ineq. (\ref{nier1})
and rearrange all terms $E_{AB}^{(i)}$ and $E_{AC}^{(j)}$ appearing there in the following way.
From any of $M-1$ Bell expressions $I_{AB}^{M,d}$ we remove
a different $E^{(i)}_{AB}$ and replace it by the corresponding 
$E^{(i)}_{AC}$ taken from $I_{AC}^{M,d}$. The removed components 
we group with the one left in $I_{AC}^{M,d}$, i.e., $E_{AC}^{(M)}$.
This allows us to rewrite the left-hand side of
Ineq. (\ref{nier1}) as
%
%
\begin{equation}\label{nier3}
(M-1)I_{AB}^{M,d}+I^{M,d}_{AC}=\sum_{j=1}^{M}\widetilde{I}_j,
\end{equation}
where for any $j=1,\ldots,M$, 
%
\begin{eqnarray}
\widetilde{I}_j&=&I_{AB}^{M,d}-E_{AB}^{(j)}+E_{AC}^{(j)}\nonumber\\
&=&\sum_{\alpha=1}^{j-1}E_{AB}^{(\alpha)}
+E_{AC}^{(j)}+\sum_{\alpha=j+1}^{M}E_{AB}^{(\alpha)}.
\end{eqnarray}

Now, with the aid of the fact that for any variable 
$\Omega$, $\sr{\Omega}+\sr{-\Omega-1}=d-1$, we can straghtforwardly see that
\begin{eqnarray}
\hspace{-3cm}&&\sum_{i=1}^{j-1}(\sr{A_i-C_j}+\sr{C_j-A_i-1})\nonumber\\
&&+\sum_{i=j+2}^{M}(\sr{C_j-A_i}+\sr{A_i-C_j-1})-(M-2)(d-1). \nonumber\\
\end{eqnarray}
is an expression that amounts to zero. By adding it to $\widetilde{I}_j$ and rearranging some terms, 
we can express $\widetilde{I}_j$ in the following way
%
\begin{eqnarray}\label{rown1}
\widetilde{I}_j&=&\sum_{i=1}^{j-1}
E_{AB}^{(i)}+\sr{C_j-A_1-1}\nonumber\\
&&+\sum_{i=2}^{j-1}(\sr{A_i-C_j}
+\sr{C_j-A_i-1})+\sr{A_j-C_j}\nonumber\\
&&+\hspace{-0.2cm}\sum_{i=j+1}^{M}E_{AB}^{(i)}+\sr{C_j-A_{j+1}}\nonumber\\
&&+\sum_{i=j+2}^{M}\hspace{-0.1cm}(\sr{C_j-A_i}+\sr{A_i-C_j-1})+\sr{A_1-C_j}\nonumber\\
&&-(M-2)(d-1).
\end{eqnarray}
One immediately notices that all the components containing $C$ 
can be written in a simpler form, allowing 
to rewrite Eq. (\ref{rown1}) as
\begin{eqnarray}\label{rown2}
\widetilde{I}_j&=&\sum_{i=1}^{j-1}E_{AB}^{(i)}+\sum_{i=1}^{j-1}(\sr{A_{i+1}-C_j}+\sr{A_i-C_j-1})\nonumber\\
&=&\sum_{i=j+1}^{M}E_{AB}^{(i)}+\sum_{i=j+1}^{M}(\sr{C_j-A_i}+\sr{A_{i+1}-C_j-1})\nonumber\\
&&-(M-2)(d-1).
\end{eqnarray}
In the last step it is enough to use the definition 
of $E_{AB}^{(i)}$ [cf. Eq. (\ref{def})], which after rearranging
terms leads us to
\begin{eqnarray}\label{row2}
\widetilde{I}_j&=&\sum_{i=1}^{j-1}(\sr{A_i-B_i}+\sr{B_i-A_{i+1}}\nonumber\\
&&\hspace{1cm}+\sr{A_{i+1}-C_j}+\sr{C_j-A_i-1})\nonumber\\
&&+\sum_{i=j+1}^{M}(\sr{A_i-B_i}+\sr{B_i-A_{i+1}}\nonumber\\
&&\hspace{1.5cm}+\sr{A_{i+1}-C_j-1}+\sr{C_j-A_i})\nonumber\\
&&-(M-2)(d-1).
\end{eqnarray}
The expressions that appear under both sums in 
Eq. (\ref{row2}) are basically the same as those in 
Eq. (\ref{BKPNd}) with $M=2$. However, they 
are ``distributed'' among three parties instead of two, and 
$B$ and $C$ have only a single measurement at their choice. It was shown in 
Ref. \cite{Marcin} that a minimization of such an expression over nonsignalling correlations
can only give its local bound, which in this case is $d-1$. Consequently, we can 
bound $\widetilde{I}_j$ over nonsignalling correlations from below as
$\widetilde{I}_j\geq (j-1)(d-1)+(M-j)(d-1)-(M-2)(d-1)=(M-1)(d-1)-(M-2)(d-1)=d-1$.
Putting this lower bound into (\ref{nier3}), one directly obtains 
Ineq. (\ref{nier1}), which completes the proof.
\end{proof}

Let us now characterize our new monogamy relations. First, it follows 
from Eqs. (\ref{nier1}) and (\ref{nier2}) that if for one of the pairs $AB$ or $AC$, say $AB$,
violates the inequality (\ref{BKPNd}), i.e., $I_{AB}^{M,d}<d-1$, then the other pair
does not violate it because $I_{AC}^{M,d}> (M-1)(d-1)$. In other words, the two pairs 
cannot simultaneously violate the BKP Bell inequality.
In particular, if $AB$ violates (\ref{BKPNd}) maximally, in which case the first inequality (\ref{nier1})
should be used, the value of $I^{M,d}$ for the other pair obeys $I_{AC}^{M,d}\geq M(d-1)$. 
This also means that our monogamies impose stronger bounds on $I_{AB}^{M,d}$
and $I_{AC}^{M,d}$ than the monogamy (\ref{MonogNtight}) derived in Ref. \cite{RaviMon1}.
%
%
In fact, our monogamy relations are tight in the sense that 
for any pair of $I_{AB}^{M,d}$ and $I_{AC}^{M,d}$ for which 
the monogamy is saturated one can find three-partite no-signaling 
correlations realizing these values. This means that one cannot 
construct a stronger monogamy relation for nonsignaling correlations
that would contain only the two values $I_{AB}^{M,d}$ and $I_{AC}^{M,d}$. 

To prove the tightness, 
let us consider inequality (\ref{nier1}) and 
three-partite correlations $\mathbf{p}_{ABC}=\{p(a,b,c|x,y,z)\}$ 
defined as
\begin{equation}
p(a,b,c|x,y,z)=\widetilde{p}_{AB}^{q}(a,b|x,y)p_C(c|z)
\end{equation}
where $\widetilde{\mathbf{p}}^{q}=\{\widetilde{p}^q(a,b,|x,y)\}$ is a mixture of the extremal 
no-signaling correlations $\mathbf{p}^{\mathrm{NL}}_2$ maximally violating (\ref{BKPNd}) 
and the local correlations $\mathbf{p}_2^{\mathrm{loc}}$ saturating it (see Sec. \ref{Sec:Preliminaries3} for the definitions), 
i.e., 
\begin{equation}
\widetilde{p}_{AB}^{q}(a,b|x,z)=qp_{AB}^{NL}(a,b|x,y)+(1-q)p_{AB}^{\mathrm{loc}}(a,b|x,y)
\end{equation}
with $q\in [0,1]$, and $p_C(c|z)$ is
the probability distribution such that $p_C(0|z)=1$ for any $z$. One directly verifies that 
for these correlations, $I_{AB}^{M,d}=(1-q)(d-1)$ and $I_{AC}^{M,d}=qM(d-1)+(1-q)(d-1)$, meaning that 
$\mathbf{p}_{ABC}$ recovers the line $(M-1)I_{AB}^{M,d}+I_{AC}^{M,d}=M(d-1)$ for 
%
%
$I_{AB}^{M,d}\leq I_{AC}^{M,d}$.
For inequality (\ref{nier2}) one distributes $\widetilde{\mathbf{p}}^q$ between $A$ and $C$ and follows the same reasoning.

%
%

It is also worth checking whether other Bell inequalities with 
the contradiction number\footnote{Recall that contradiction number of 
a given Bell inequality is the smallest set of measurements of, say $B$, whose removal 
trivializes it, i.e., its maximal nonsignaling violation becomes achievable by local correlations \cite{RaviMon1}.} one would obey monogamy (\ref{nier1}) and (\ref{nier2}).
This is certainly not the case. As a counterexample let us consider a simple modification of $I_{AB}^{3,2}$ given by
\begin{eqnarray}
I_{AB}'&:=&\sr{A_0-B_0}+\sr{B_0-A_1}+\sr{A_1-B_1}\nonumber\\
&&+2(\sr{B_1-A_2}+\sr{A_2-B_2})+\sr{B_2-A_0-1}\nonumber\\
&\geq& 1.
\end{eqnarray}
The maximal nonsignaling violation of this Bell inequality corresponds to 
$I'_{AB}=0$. Exploiting linear programming it is straightforward to 
verify that the tight monogamy relations this Bell inequality obeys are
\begin{equation}
\left\{
\begin{array}{c}
3I'_{AB}+I'_{AC}\geq 4\\[1ex]
I'_{AB}+3I'_{AC}\geq 4.
\end{array}
\right.
\end{equation}

A stronger counterexample can be constructed using
the $\mathcal{I}^{3322}$ Bell inequality \cite{Sliwa}:
\begin{eqnarray}
\mathcal{I}^{3322}&:=&\langle A_1B_1\rangle+\langle A_1B_2\rangle+\langle A_2B_1\rangle+
\langle A_2B_2\rangle-\langle A_1B_3\rangle\nonumber\\
&&+\langle A_2B_3\rangle-\langle A_3B_1\rangle+
\langle A_3B_2\rangle\nonumber\\
&&-\langle A_2\rangle-\langle B_1\rangle-2\langle B_2\rangle\leq 0
\end{eqnarray}
First, let us notice that the maximal no-signaling value of
this Bell expression is $\mathcal{I}^{3322}=4$. 
Now, assuming that $AC$ violate maximally the above Bell inequality, 
i.e., $\mathcal{I}^{3322}_{AC}=4$, it is easy to verify with the aid of the 
linear programming that $AB$ can maximally violate $\mathcal{I}^{3322}_{AB}$ too. 
In other words, there exists a three-partite nonsignaling
correlations $\{p(abc|xyz)\}$ allowing to simultaneously violate both $\mathcal{I}^{3322}_{AB}$
and $\mathcal{I}^{3322}_{AC}$ maximally. Thus, $\mathcal{I}^{3322}$ is an example 
of a Bell inequality for which it is not possible to
formulate a monogamy of the form (\ref{nier1}) and (\ref{nier2}).

\section{Generalization to the multipartite case}

The monogamy relation (\ref{ThreeMonogamy}) can be straightforwardly generalized 
to an arbitrary number of parties $N$. Let us prove the following theorem.
\begin{thm}For any $(N+1)$-partite probability distribution 
$\{p(\mathbf{a}|\mathbf{x})\}$ with $M$ measurements and $d$ outcomes per site, 
the following inequalities are satisfied
\begin{eqnarray}
\label{nierN1} \hspace{-0.4cm}(M-1)I_{A_1\ldots A_N}^{N,M,d}+I^{N,M,d}_{A_1\ldots A_{N-1}A_{N+1}}\geq M^{N-1}(d-1), \\[1ex]
\label{nierN2} \hspace{-0.4cm} I_{A_1\ldots A_{N}}^{N,M,d}+(M-1)I^{N,M,d}_{A_1\ldots A_{N-1}A_{N+1}}\geq M^{N-1}(d-1),
\end{eqnarray}
where the first inequality is applied when $I_{A_1\ldots A_N}^{N,M,d}\leq I_{A_1\ldots A_{N-1}A_{N+1}}^{N,M,d}$, while 
the second one in the opposite case.
\end{thm}
\begin{proof}
Let us first assume that $N$ is even and show that 
\begin{widetext}
\begin{eqnarray}\label{Gargalo}
&&\hspace{-1cm}\sum_{\boldsymbol{\alpha}}
\sr{A_{\alpha_1}^{(1)}-A_{\alpha_1+\alpha_2-1}^{(2)}+\ldots
+(-1)^{N}A_{\alpha_{N-2}+\alpha_{N-1}-1}^{(N-1)}+(-1)^{N-1}A_{\alpha_{N-1}}^{(N)}}\nonumber\\
&&\hspace{-1cm}=\sum_{\boldsymbol{\alpha}}
\sr{A_{\alpha_1+1}^{(1)}-A_{\alpha_1+\alpha_2-1}^{(2)}+\ldots
+(-1)^{N}A_{\alpha_{N-2}+\alpha_{N-1}-1}^{(N-1)}+(-1)^{N-1}A_{\alpha_{N-1}+1}^{(N)}},
\end{eqnarray}
that is, we want to show that the expression in the first line of the above equation
does not change if we add one to both indices $\alpha_1$ and $\alpha_{N-1}$. To this end, we first 
shift $\alpha_{N-1}\to \alpha_{N-1}+1$, which yields  
\begin{eqnarray}\label{sum1}
&&\hspace{-0.8cm}\sum_{\boldsymbol{\alpha}}
\sr{A_{\alpha_1}^{(1)}-A_{\alpha_1+\alpha_2-1}^{(2)}+\ldots
+(-1)^{N}A_{\alpha_{N-2}+\alpha_{N-1}-1}^{(N-1)}+(-1)^{N-1}A_{\alpha_{N-1}}^{(N)}}\nonumber\\
&&\hspace{-0.6cm}=\sum_{\boldsymbol{\alpha}'}\sum_{\alpha_{N-1}=0}^{M-1}
\sr{A_{\alpha_1+1}^{(1)}-A_{\alpha_1+\alpha_2-1}^{(2)}+\ldots
+(-1)^{N}A_{\alpha_{N-2}+\alpha_{N-1}}^{(N-1)}+(-1)^{N-1}A_{\alpha_{N-1}+1}^{(N)}},
\end{eqnarray}
where $\boldsymbol{\alpha}'=\alpha_1,\ldots,\alpha_{N-2}$.
Then, by employing the rule that $A_{M+\alpha}^{(i)}=A_{\alpha}^{(i)}+1$, 
one can show that in the above sum 
the term corresponding to $\alpha_{N-1}=0$ is the same as the one corresponding to 
$\alpha_{M-1}$, meaning that we can rewrite (\ref{sum1}) as
\begin{eqnarray}
&&\hspace{-0.8cm}\sum_{\boldsymbol{\alpha}}
\sr{A_{\alpha_1}^{(1)}-A_{\alpha_1+\alpha_2-1}^{(2)}+\ldots
+(-1)^{N}A_{\alpha_{N-2}+\alpha_{N-1}-1}^{(N-1)}+(-1)^{N-1}A_{\alpha_{N-1}}^{(N)}}\nonumber\\
&&\hspace{-0.6cm}=\sum_{\boldsymbol{\alpha}'}\sum_{\alpha_{N-1}=1}^{M}
\sr{A_{\alpha_1+1}^{(1)}-A_{\alpha_1+\alpha_2-1}^{(2)}+\ldots
+(-1)^{N}A_{\alpha_{N-2}+\alpha_{N-1}}^{(N-1)}+(-1)^{N-1}A_{\alpha_{N-1}+1}^{(N)}}.
\end{eqnarray}
We then shift the last but one index $\alpha_{N-2}\to \alpha_{N-2}-1$, which gives
\begin{eqnarray}
&&\hspace{-0.5cm}\sum_{\boldsymbol{\alpha}}
\sr{A_{\alpha_1}^{(1)}-A_{\alpha_1+\alpha_2-1}^{(2)}+\ldots
+(-1)^{N}A_{\alpha_{N-2}+\alpha_{N-1}}^{(N-1)}+(-1)^{N-1}A_{\alpha_{N-1}+1}^{(N)}}\nonumber\\
&&\hspace{-0.5cm}=\sum_{\boldsymbol{\alpha}''}\sum_{\alpha_{N-2}=2}^{M+1}
\sr{A_{\alpha_1+1}^{(1)}-A_{\alpha_1+\alpha_2-1}^{(2)}+\ldots+(-1)^{N-1}A_{\alpha_{N-3}+\alpha_{N-2}-2}^{(N-1)}\nonumber\\
&&\hspace{3cm}+(-1)^{N}A_{\alpha_{N-2}+\alpha_{N-1}-1}^{(N-1)}+(-1)^{N-1}A_{\alpha_{N-1}+1}^{(N)}}.
\end{eqnarray}
where $\boldsymbol{\alpha}''=\boldsymbol{\alpha}\setminus \alpha_{N-2}=\alpha_1,\ldots,\alpha_{N-3},\alpha_{N-1}$. 
Denoting $A^{(i)}_{M}=A^{(i)}_0+1$, we can rewrite the above as
\begin{eqnarray}
&&\hspace{-0.5cm}\sum_{\boldsymbol{\alpha}}
\sr{A_{\alpha_1}^{(1)}-A_{\alpha_1+\alpha_2-1}^{(2)}+\ldots
+(-1)^{N}A_{\alpha_{N-2}+\alpha_{N-1}}^{(N-1)}+(-1)^{N-1}A_{\alpha_{N-1}+1}^{(N)}}\nonumber\\
&&\hspace{-0.5cm}=\sum_{\boldsymbol{\alpha}''}\sum_{\alpha_{N-2}=1}^{M}
\sr{A_{\alpha_1+1}^{(1)}-A_{\alpha_1+\alpha_2-1}^{(2)}+\ldots+(-1)^{N-1}A_{\alpha_{N-3}+\alpha_{N-2}-2}^{(N-1)}\nonumber\\
&&\hspace{3cm}+(-1)^{N}A_{\alpha_{N-2}+\alpha_{N-1}-1}^{(N-1)}+(-1)^{N-1}A_{\alpha_{N-1}+1}^{(N)}}.
\end{eqnarray}
\end{widetext}
One then repeats this procedure with the remaining $\alpha_i$'s, following the rule
$\alpha_{N-i}\to \alpha_{N-i}+(-1)^{i-1}$ with $i=1,\ldots,N-1$ and uses the above reasoning 
to obtain Eq. (\ref{Gargalo})).

Having Eq. (\ref{Gargalo}), we then introduce a set of auxiliary variables
\begin{equation}
X_{\alpha_{N-1}}^{\boldsymbol{\alpha}'}=A_{\alpha_1+1}^{(1)}-A_{\alpha_1+\alpha_2-1}^{(2)}+\ldots
+(-1)^{N}A_{\alpha_{N-2}+\alpha_{N-1}-1}^{(N-1)}
\end{equation}
and note that for any choice of $\boldsymbol{\alpha}'$,
the outcomes of $X_{\alpha_{N-1}}^{\boldsymbol{\alpha}'}$ belong to $\{0,\ldots,d-1\}$. Therefore, 
the Bell expression $I_{A_1\ldots A_N}^{N,M,d}$, which we rewrite as
\begin{eqnarray}
I_{A_1\ldots A_N}^{N,M,d}&=&\sum_{\boldsymbol{\alpha}'}\sum_{\alpha_{N-1}=1}^{M}
\hspace{-0.2cm}\left(\sr{A^{(N)}_{\alpha_{N-1}}-X^{\boldsymbol{\alpha}'}_{\alpha_{N-1}}}+\sr{X^{\boldsymbol{\alpha}'}_{\alpha_{N-1}}-A^{(N)}_{\alpha_{N-1}+1}}\right)\nonumber\\
&=&\sum_{\boldsymbol{\alpha}'}I_{XA^{(N)}}^{(\boldsymbol{\alpha}')},
\end{eqnarray}
can be expressed as a combination of 
$M^{N-2}$ bipartite Bell expressions between the auxiliary variable 
$X^{\boldsymbol{\alpha}'}$ and $A^{(N)}$. Clearly, each $I_{XA^{(N)}}^{\boldsymbol{\alpha}'}$ 
obeys the monogamy relations in Eqs. (\ref{nier1}) and (\ref{nier2}).

Now, we plug the above form into (\ref{nierN1}) which yields and use
(\ref{nier1}) to obtain
\begin{eqnarray}
&&\hspace{-1cm}(M-1)I_{A_1\ldots A_N}^{N,M,d}+I_{A_1\ldots A_{N-1}A_{N+1}}^{N,M,d}\nonumber\\
&&\hspace{0.5cm}=\sum_{\boldsymbol{\alpha}'}\left[(M-1)I_{XA^{(N)}}^{(\boldsymbol{\alpha}')}+I_{XA^{(N+1)}}^{\boldsymbol{\alpha}'}\right]\nonumber\\
&&\hspace{0.5cm}\geq\sum_{\boldsymbol{\alpha}'}\left[M(d-1)\right]=M^{N-1}(d-1),
\end{eqnarray}
which gives inequality (\ref{nierN1}). In exactly the same way one proves (\ref{nierN2}).

To prove the monogamy relations for odd $N$, one first shows, using 
the same reasoning as above, that 
\begin{eqnarray}
&&\hspace{-0.7cm}\sum_{\boldsymbol{\alpha}}
\sr{A_{\alpha_1+\alpha_2-1}^{(2)}-A_{\alpha_1+1}^{(1)}+\ldots\nonumber\\
&&+(-1)^{N-1}A_{\alpha_{N-2}+\alpha_{N-1}-1}^{(N-1)}+(-1)^{N}A^{(N)}_{\alpha_{N-1}}},\nonumber\\
&&\hspace{-0.7cm}=\sum_{\boldsymbol{\alpha}}
\sr{A_{\alpha_1+\alpha_2-1}^{(2)}-A_{\alpha_1}^{(1)}+\ldots\nonumber\\
&&+(-1)^{N-1}A_{\alpha_{N-2}+\alpha_{N-1}-1}^{(N-1)}+(-1)^{N}A^{(N)}_{\alpha_{N-1}+1}}.
\end{eqnarray}
Then, one defines additional variables
\begin{equation}
\widetilde{X}_{\alpha_{N-1}}^{(\boldsymbol{\alpha}')}=A_{\alpha_1+\alpha_2-1}^{(2)}-A_{\alpha_1}^{(1)}+\ldots
+(-1)^{N-1}A_{\alpha_{N-2}+\alpha_{N-1}-1}^{(N-1)},
\end{equation}
which allows one to rewrite the Bell expression as
\begin{eqnarray}
I_{A_1\ldots A_N}^{N,M,d}&=&\sum_{\boldsymbol{\alpha}'}\sum_{\alpha_{N-1}=1}^{M}\left(\sr{A^{(N)}_{\alpha_{N-1}}-\widetilde{X}^{(\boldsymbol{\alpha}')}_{\alpha_{N-1}}}\right.\nonumber\\
&&\hspace{2cm}\left.+\sr{\widetilde{X}^{(\boldsymbol{\alpha}')}_{\alpha_{N-1}}-A^{(N)}_{\alpha_{N-1}+1}}\right)\nonumber\\
&=&\sum_{\boldsymbol{\alpha}'}\widetilde{I}_{XA^{(N)}}^{(\boldsymbol{\alpha}')}.
\end{eqnarray}
Plugging the above formula into the left-hand side of (\ref{nierN1}) and exploiting the fact that $\widetilde{I}_{XA^{(N)}}^{(\boldsymbol{\alpha}')}$ satisfies the monogamy relation (\ref{nier1}), one obtains (\ref{nierN1}). In the same way one can prove (\ref{nierN2}). This completes the proof.
\end{proof}

As before, it follows from (\ref{nierN1}) and (\ref{nierN2}) that if an $N$-partite subset 
of $N+1$ parties $A_1,\ldots, A_{N+1}$ violates the inequality (\ref{BKPNMd}), i.e., 
$I_{A_1\ldots A_N}^{N,M,d}<M^{N-2}(d-1)$, then any other $N$-partite subset cannot violate it. 
Moreover, the monogamy relations (\ref{nierN1}) and (\ref{nierN2})
are tight for any number of parties $N$. To demonstrate it let us consider $(N+1)$-partite correlations $\mathbf{p}_{N+1}$
given by the following formula
\begin{eqnarray}
&&p(a_1,\ldots, a_{N+1}|x_1,\ldots, x_{N+1})\nonumber\\
&&=
\widetilde{p}_{A_1\ldots A_N}^{q}(a_1,\ldots, a_N|x_1,\ldots, x_N)p_{A_{N+1}}(a_{N+1}|x_{N+1})\nonumber\\
\end{eqnarray}
where $\widetilde{\mathbf{p}}^{q}$ is a mixture of the extremal 
no-signaling correlations $\mathbf{p}_N^{\mathrm{NL}}$
maximally violating the Bell inequality (\ref{BKPNMd}) and the classical correlations $\mathbf{p}^{\mathrm{loc}}$
for which this inequality is saturated:
\begin{eqnarray}
&&\hspace{-0.7cm}\widetilde{p}_{A_1\ldots A_N}^{q}(a_1,\ldots, a_N|x_1,\ldots, x_N)\nonumber\\
&&=
qp_{A_1\ldots A_N}^{\mathrm{NL}}(a_1,\ldots, a_N|x_1,\ldots, x_N)\nonumber\\
&&\hspace{0.5cm}+(1-q)p_{A_1\ldots A_N}^{\mathrm{loc}}(a_1,\ldots, a_N|x_1,\ldots, x_N)
\end{eqnarray}
with $q\in [0,1]$, and $p_{A_N}(a_{N+1}|x_{N+1})$ being 
the probability distribution defined through $p_{A_{N+1}}(0|x_{N+1})=1/d$ for any $x_{N+1}$.

\section{Application to guessing probability}

Let us now link our new monogamy relations to 
the guessing probability, which is a central 
quantity in device-independent tasks such as 
randomness certification \cite{gen2}. For this purpose,
we demonstrate that for the case of $M=2$ the inequalities (\ref{ThreeMonogamy}) 
imply the monogamy relations derived in Ref. \cite{Augusiak}, which in turn
were shown to impose tight bounds on the guessing probability.

Let us begin with the definition of the guessing probability. Consider bipartite correlations 
$\mathbf{p}_{\mathrm{obs}}=\{p_{\mathrm{obs}}(a,b|x,y)\}$ with $M$ measurement and $d$ outcomes that the parties 
$A$ and $B$ observe in a Bell experiment. Due to the fact that the set of
correlations obeying the no-signaling principle is a convex polytope with 
a finite number of vertices, $\mathbf{p}_{\mathrm{obs}}$ can always be 
written as a convex combination, i.e., 
\begin{equation}\label{decomposition}
p_{\mathrm{obs}}(a,b|x,y)=\sum_e p(e|x,y) {p}^e_{\mathrm{ex}}(a,b|x,y),
\end{equation}
where $\mathbf{p}^e_{\mathrm{ex}}=\{p^e_{\mathrm{ex}}(a,b|x,y)\}$ are the vertices of the no-signaling polytope
and $p(e|x,y)$ is a probability distribution that takes into account possible
correlations between $e$ and the measurements $x$ and $y$.
The probability of guessing the outcome of 
the $x$th measurement of $A$ is then defined as
\begin{equation}\label{decomposition2}
G(x,\mathbf{p}_{\mathrm{obs}})=\max \sum_{e}p(e|x,y) G(x,\mathbf{p}_{\mathrm{ex}}^e),
\end{equation}
where the maximum is taken over all decompositions (\ref{decomposition}) 
in which and $G$ for an extremal box is simply $G(x,\mathbf{p}_{\mathrm{ex}})=\max_a p_{\mathrm{ex}}(a|x)$
with $p_{\mathrm{ex}}(a|x)$ being the marginal of $p_{\mathrm{ex}}(a,b|x,y)$ corresponding to the party $A$.
In a fully analogous way one defines the guessing probability for $B$'s measurements.

Let us now show that the monogamy relations (\ref{ThreeMonogamy}) imply tight bounds on $G$.
To this end, we consider the monogamy (\ref{ThreeMonogamy}) for $M=2$, i.e., 
%
$I_{AB}^{2,d}+I_{AC}^{2,d}\geq 2(d-1).$
%
Now, in $I^{2,d}_{AC}$ let us set $C_2=C_1+1$ and use the fact that 
$\sr{C_1-A_2}+\sr{A_2-C_1-1}=d-1$, which yields
\begin{equation}
I_{AB}^{2,d}+\sr{A_1-C_1}+\sr{C_1-A_1}\geq d-1.
\end{equation}
In a similar way one derives analogous inequalities for the remaining
pairs of measurements by $A$ and $C$, which finally gives
\begin{equation}
I_{AB}^{2,d}+\sr{A_i-C_j}+\sr{C_j-A_i}\geq d-1
\end{equation}
with $i,j=1,2$. These monogamy relations were derived in Ref. \cite{Augusiak}, and, by exploting
the properties of the expectation value $\sr{\cdot}$, 
they can be rewritten in a more meaningful way as
\begin{equation}\label{inequalities}
I_{AB}^{2,d}+1\geq dp(A_i=C_j),
\end{equation}
where $p(A_i=C_j)$ is the probability that $A$ and $C$ obtain the same outcomes
upon performing $i$th and $j$th measurement, respectively. Thus, the latter 
can be understood as a measure of how outcomes
of the measurements by $A$ and $C$ are (classically) correlated, 
these inequalities relate the nonlocality observed by $A$ and $B$ to the 
knowledge that $C$ can have about the outcomes of any measurements
by $A$ and $B$. Thus, inequalities (\ref{inequalities}) are naturally 
related to the guessing probability, and, in fact, it was shown in Ref. \cite{Augusiak}
that for any probability distribution $\mathbf{p}=\{p(a,b|x,y)\}$ 
they yield the following tight bounds
\begin{equation}\label{bound}
\max_a p_A(a|x)\leq \frac{1}{d}\left[1+I_{AB}^{2,d}(\mathbf{p})\right].
\end{equation}
The same inequality holds for the marginals of $B$. 

Now, in order to obtain a bound on the guessing probability $G$, let us 
again consider a probability distribution $\mathbf{p}_{\mathrm{obs}}$
observed by $A$ and $B$. Let us then
apply (\ref{bound}) to every element in the decomposition
(\ref{decomposition2}), which results in 
\begin{equation}
G(x,\mathbf{p}_{\mathrm{obs}})\leq \frac{1}{d}\max \sum_{e}p(e|x,y)\left[1+I^{2,d}_{AB}(\mathbf{p}_{\mathrm{ex}}^e)\right]
\end{equation}
Exploiting then the fact that the Bell expression is linear in $\mathbf{p}$ and
the decomposition (\ref{decomposition}), one finally obtains
\begin{equation}
G(x,\mathbf{p}_{\mathrm{obs}})\leq \frac{1}{d}\left[1+I^{2,d}_{AB}(\mathbf{p}_{\mathrm{obs}})\right].
\end{equation}
Due to the fact that the monogamy relations (\ref{ThreeMonogamy}) are tight, 
the above bound is tight too. This means that it is not possible to 
find a better bound on the guessing probability in terms of only the 
value of the Bell expression $I^{2,d}_{AB}$ for a given probability distribution.
In particular, at the point of maximal nonsignaling violation, $I^{M,d}=0$, 
the above bound implies that $G(x,\mathbf{p}_{\mathrm{obs}})=1/d$, meaning that 
the outcome of any measurement of $A$ is completely random, while 
for $I^{M,d}=d-1$, $G(x,\mathbf{p}_{\mathrm{obs}})\leq 1$.

\section{Conclusion}

Monogamy relations are intriguing properties of no-signaling physical theories. 
They tell us how correlations can be distributed among distant parties.
Apart from the fundamental interest, monogamy relations are also 
of importance for device-independent information tasks. 

There have been some constructions of monogamy relations for no-signaling correlations, 
however, they are either complicated or not tight. In this note
we have introduced simple bipartite monogamy relations are 
then generalized to any number of observers.
To formulate them, we have exploited the
BKP Bell inequality introduced in Ref. \cite{BKP} and its generalization 
to the multipartite case given in Ref. \cite{rodrigo}. Importantly, the resulting
trade-offs are not only simple, but also tight. Moreover, in certain scenarios they are 
shown to impose tight bounds on the guessing probability, which is a measure of 
how random are outcomes of measurements in a Bell scenario. 

We leave as an problem as to whether the above method can be generalized to 
other Bell inequalities. Or, more concretely, for which Bell inequalities 
one can formulate a simple and tight monogamy relation of the form (\ref{General}).


\begin{acknowledgments}Discussions with A. Ac\'in, M. Demianowicz, M. Paw\l{}owski, R. Ramanathan
and J. Tura are acknowledged.
This project has received funding from the European Union's Horizon 2020 
research and innovation programme under the Marie Sk\l{}odowska-Curie grant agreement No 705109.
\end{acknowledgments}


\end{document}